\documentclass[reprint,pra]{revtex4-1}

\usepackage{amsmath}
\usepackage{amsfonts}
\usepackage{color}
\usepackage[utf8]{inputenc}
\usepackage[dvips]{graphicx}
\usepackage{tikz}
\usepackage{indentfirst}
\usepackage{fancyhdr}
\usepackage{float}
\usepackage{subfigure}
\usepackage{amsthm, amssymb,bbm}
\newtheorem{theorem}{Theorem}
\newtheorem*{theorem*}{Theorem}
\newtheorem{corollary}{Corollary}
\newtheorem{prop}{Proposition}

\begin{document}
\title{On the uniqueness of the steady-state solution of the Lindblad-Gorini-Kossakowski-Sudarshan equation.}
\author{Davide Nigro}
\affiliation{Dipartimento di Fisica dell’Università di Pisa and I.N.F.N.\\ Sezione di Pisa, Largo Pontecorvo 3, I-56127 Pisa, Italy}
\begin{abstract}
The aims of this paper are two. The first is to give a brief review of the most relevant theoretical results concerning the uniqueness of the steady-state solution of the Lindblad-Gorini-Kossakowski-Sudarshan master equation and the criteria which guarantee relaxingness and irreducibility of dynamical semigroups. In particular, we test and discuss their physical meaning by considering their applicability to the characterisation of the simplest open quantum system \emph{i.e.} a two-level system coupled to a bath of harmonic oscillators at zero temperature. The second aim is to provide a set of sufficient conditions which guarantees the uniqueness of the steady-state solution and its attractivity. Starting from simple assumptions, we derive simple criteria that can be efficiently exploited to characterise the behavior of dissipative systems of spins and bosons (with truncated Fock space), and a wide variety of other open quantum systems recently studied. 
\end{abstract}
\maketitle

\section{Introduction}
During the last few decades systems of cold atoms, molecules and trapped ions, as well as photons, have been extensively used in so many fields of research that they seem to represent the key to understand and investigate the most fundamental laws of quantum mechanics \cite{Intro3,haroche,Intro1,Intro2,carusotto}. Indeed, thanks to the impressive progress which has been made especially in optics, it is possible to use these systems to engineer microscopic Hamiltonians for the purpose of quantum simulation \cite{Intro3b,hartmann,noh}: by using lasers one can achieve those critical regimes at which the system dynamics becomes completely ruled by principles of quantum mechanics. However, keeping systems in the proper working conditions 
can be quite a hard task: the coupling between the microscopic degrees of freedom and the sorroundings naturally leads to losses. If these losses are negligible, the quantum system can be approximate as an ideal and closed system which follows the standard quantum theory. If it is not possible to neglect the leakage induced by the environment one needs to change point of view and use the formalism of \emph{open} quantum systems.\\
The first examples of a rigorous treatment of the open quantum system dynamics were given during the 70s by Davies while considering a harmonic oscillator \cite{daviesoscillator} and a N-level atom \cite{daviesnlevel} coupled to a heat bath and by Pulè \cite{pule} in the context of a single spin coupled to an infinite bath of harmonic oscillators. During the same years, Kossakowski and Ingarden \cite{kossakowski,ingarden}, while attempting to provide a mathematical framework suitable for the description of the irreversible dynamics proper of the open quantum systems, introduced the concept of \emph{quantum dynamical semigroup}: while in the case of closed systems it is possible to move both ``forward" and ``backwards" in time by using unitary operators which are elements of a one-parameter group, in the case of quantum open systems time-translations form a semigroup \emph{i.e.} elements do not have an inverse.\\
In 1976, Lindblad \cite{Lindblad} and Gorini, Kossakowski and Sudarshan \cite{Gorini} derived the general form of the generator of completely positive quantum dynamical semigroups and introduced the master equation nowadays known as the LGKS master equation which prescribes the time-evolution of a open quantum system weakly coupled to a Markovian environment. Between 1976 and 1978 a series of theoretical papers concerning the \emph{uniqueness} of the solution of the LGKS master equation appeared in literature. In \cite{spohn1} and \cite{spohn2} Spohn gave a set of sufficient conditions for uniqueness based respectively on the properties of the decay-rate matrix and on the properties of the Lindblad operators entering in the generator of the dynamical semigroup.
In \cite{frigerio1} and \cite{frigerio2} Frigerio considered under which sufficient conditions a dynamical semigroup possessing a faithful normal stationary state admits a unique equilibrium state. Frigerio in \cite{frigerio1} also derived, for this particular class of problems, the equivalence between \emph{irreducibility } and the uniqueness of the equilibrium state. Evans in \cite{evans} provided a necessary and suffiecient condition for irreducibility based on the study of von Neumann algebras.\\
 After three decades, the renewed and growing interest for the physics of open quantum systems has led to the birth of two completely new lines of research. The first deals with the development of tools and techniques suitable for the purposes of stabilizing the exotic quantum states observed in equilibrium quantum systems and morever suitable for the synthesis of systems possessing those features required by quantum information protocols. These goals can be achieved combining reservoir engineering with control techniques \cite{jin,biella,lebreuilly,diehl,bardyn,sauer,dengis}. The second line deals with the study and the characterisation of those phenomena proper of the open quantum dynamics. Of particular interest are the transport properties in driven-dissipative systems \cite{biella2,fitzpatrick,debnath} and the analysis of dissipative phase transitions \cite{lee,jin2,ciuti5,rota}, which can be considered as an extension to the open quantum world of the standard quantum phase transitions.\\
Though these two research areas deal with different aspects of the physics of open quantum systems, they do share the need for a precise knowledge about the asymptotic behavior of open quantum systems \emph{i.e.} the properties encoded in the steady-states. 
Such properties and their origin have been extensively investigated during the last decade. Several new interesting results concerning the general structure of Hilbert spaces and the characterisation of the time-evolution of open quantum systems, as well as the role of symmetries and conserved quantities have been published  \cite{baum1,baum2,ticozzi1,ticozzi2,schirmer,baum3,albert,albert2}. 
Nevertheless, as far as we know, the more practical criteria for the characterisation of the asymptotic properties of open quantum systems are those discussed in \cite{spohn1,spohn2,frigerio1,frigerio2,evans}.
\newline
The purposes of this paper are two. The first is to review the classical papers cited above: 
since they provide powerful criteria in an extremely formal way, we do believe that a detailed discussion about their content would be helpful for our community. Our second goal is to provide a simple set of sufficient conditions which guarantees the uniqueness of the steady-state solution. As we will show, they apply to a wide variety of open quantum systems that have been subject of intense studies during the last decade.\\
The paper is organised as follows. In section \ref{sec:theory} we briefly review the main ideas of the formalism of dynamical semigroups, introducing the LGKS equation and its dual expression. In section \ref{sec:review} we review the classical results cited above concerning the uniqueness of the steady-state configurations, the conditions under which a dynamical semigroup is \emph{relaxing} and those which guarantee its \emph{irreducibility}. 
In particular, we discuss their theoretical contents by considering their applicability to the simplest prototypical model describing a open quantum system \emph{i.e.} a two-level system coupled to a bath of harmonic oscillators at zero temperature. In this way, we will show also how the uniqueness of the steady-state is related to the relaxingness and to the irreducibility of dynamical semigroups. In section \ref{sec:theorems} we discuss a new set of sufficient conditions that guarantees the uniqueness of the steady-state solution for both single and composite quantum systems. In section \ref{sec:conclusions} we summarise our results and draw our conclusions.

\section{Theory of dynamical semigroups}\label{sec:theory}
In this section, we proceed in the same way of Lindblad in \cite{Lindblad} who gave a description of the theory of dynamical semigroups starting from the formalism introduced by Kossakowski in \cite{kossakowski} for the description of non-hamiltonian systems.\\
According to the standard formalism used in quantum mechanics, to every quantum system it is possible to associate a separable Hilbert space $\mathbb{H}$ and admissible \emph{states} of our quantum system are represented by a self-adjoint positive semidefinite linear operator of unit trace \emph{i.e.} \emph{density operators}. The set of all the density operators, denoted by $\mathcal{P}(\mathbb{H})$, is a convex set: given any couple of states $\rho_1$,$\rho_2$ $\in$ $\mathcal{P}(\mathbb{H})$, also $\rho_p \equiv p\rho_1 +(1-p)\rho_2$, with $p\in \left[0,\,1\right]$ is a element of $\mathcal{P}(\mathbb{H})$ (in literature this set is also denoted as $\mathcal{T}(\mathbb{H})$, the set of \emph{trace class} operators on $\mathbb{H}$, see \emph{e.g.}  \cite{spohnrev}).
\emph{Observables} correspond to self-adjoint operators on $\mathbb{H}$ and are elements of $\mathcal{B}(\mathbb{H})$, the set of bounded operators on $\mathbb{H}$. The mean value of an observable $O$ at a state $\rho$, that is $\langle O\rangle$, can be obtained evaluating $\mbox{Tr}(O\rho)$, where $\mbox{Tr}(\cdot)$ denotes the trace operation.\\
Let $\mathcal{H}$ be the Hamiltonian of our physical system. The time-evolution of the system is given by the family of two-parameters operators $\mathbb{S}_0(\mathbb{H})=\left\{\Lambda_{t,\,s};\,t>s;\,t,\,s \in \mathbb{R}\right\}$ acting on the set $\mathcal{P}(\mathbb{H})$. $\Lambda_{t,\,s}$ has the form 
\begin{equation}\label{eq:generalgenerator}
\Lambda_{t,\,s}=T\mbox{exp}\left(\int_{s}^{t}\,L_{0}(t')dt'\right),
\end{equation}
where $T$ denotes the time ordering operator, and
\begin{equation}
L_{0}(t)\rho=-i\left[\mathcal{H},\,\rho\right];\quad \rho\in \mathcal{P}(\mathbb{H}),
\end{equation}
is the Liouville operator (also called Liouvillian) of the system ($\hbar=1$) and $\left[A,\,B\right]\equiv AB-BA$ denotes the commutator between the operators $A$ and $B$. The relation (\ref{eq:generalgenerator}) is usually written as the von Neumann equation of motion for the system state $\rho$
\begin{equation}\label{eq:vonNeumannEq}
\frac{d}{dt}(\Lambda_{t,\,s}\rho)=-i\left[\mathcal{H},\,\Lambda_{t,\,s}\rho\right]
\end{equation}
with starting condition $\Lambda_{s,\,s}\rho=\rho$, being $\Lambda_{s,\,s}=\mathbbm{1}$ the identity map. Given the system state $\rho$ at the time $s$, the system state at $t>s$ is given by $\Lambda_{t,\,s}\rho$.\\
In general, $\mathbb{S}_0(\mathbb{H})$ is a \emph{semi-group} under the following composition law
\begin{equation}
\Lambda_{t,\,s}\Lambda_{s,\,u}=\Lambda_{t,\,u};\quad t>s>u,\quad t,s,u\in\mathbb{R},
\end{equation}
A physical system for which its dynamical semi-group $\mathbb{S}_0(\mathbb{H})$ can be extended to a group $\mathbb{G}_0(\mathbb{H})$, with the introduction of the inverse operator $\Lambda^{-1}_{t,\,s}=\Lambda_{s,\,t}\,\forall t,s \in \mathbb{R}$, is called \emph{Hamiltonian system}, since $\mathbb{G}_0(\mathbb{H})$ is completely determined by the Hamiltonian $\mathcal{H}$. In this case the Hamiltonian $\mathcal{H}$ is the generator of the group. This is the case of ordinary (closed) quantum systems. A physical system for which the dynamical semi-group cannot be extended to a group is called \emph{non-Hamiltonian}. In this case we denote the semigroup by $\mathbb{S}(\mathbb{H})$. This is the case of \emph{open quantum systems}, where one has a \emph{total} system consisting of two parts, but is interested in characterising only the dynamics of one of the two subsystems. Let us call these two parts $\mathcal{R}$ and $\mathcal{S}$: $\mathcal{R}$ is the so called \emph{environment} or \emph{reservoir}, while $\mathcal{S}$ denotes the subsystem of interest, the \emph{open} part of the total system which is coupled to $\mathcal{R}$. While the total system can be considered as a Hamiltonian system, $\mathcal{S}$ suffers a time evolution which is non-Hamiltonian. This means that, if we consider a state of the total system $\rho$, the reduced density matrix $\rho_{\mathcal{S}}=\mbox{Tr}_{\mathcal{R}}\left[\rho\right]$, which describes the state of the subsystem $\mathcal{S}$ as a part of the larger system $\mathcal{S}+\mathcal{R}$, $\rho_{\mathcal{S}}$ does not evolve in time according to the von Neumann equation. 
If one considers \emph{temporarily homogeneous} semigroups (Markov approximation), that is semigroups in which all the elements $\Lambda_{t,\,s} \in \mathbb{S}(\mathbb{H})$ ($t>s$) are functions only of $t-s$ so that the composition law becomes 
\begin{equation}
\Lambda_{t}\Lambda_{s}=\Lambda_{t+s};\quad t,s\geq 0,\quad \mbox{being}\,\Lambda_{0}=\mathbbm{1},
\end{equation}
the time-evolution of the reduced density matrix $\rho_{\mathcal{S}}$ can be casted in the following form
\begin{equation}\label{eq:lindbladequation}
\frac{d}{dt}(\Lambda_{t}\rho_{\mathcal{S}})=L(\Lambda_{t}\rho_{\mathcal{S}}),\quad\Lambda_{0}\rho_{\mathcal{S}}=\rho_{\mathcal{S}},
\end{equation}
where $L$ is the generator of the dynamical semigroup $\mathbb{S}(\mathbb{H}_{\mathcal{S}})$.\\
The expression of $L$ has been independently derived in the 70s by Lindblad \cite{Lindblad} and Gorini, Kossakowski and Sudarshan \cite{Gorini}. In general, the rhs of Eq.(\ref{eq:lindbladequation}) has the following expression
\begin{equation}\label{eq:lindbladdissipator}
\begin{split}
L(\rho_{\mathcal{S}})=&-i\left[\mathcal{H}_{\mathcal{S}},\,\rho_{\mathcal{S}}\right]+ \\
&+\sum_{i \in I}\gamma_{i}\left[B_i\rho_{\mathcal{S}}B^{\dagger}_i-\frac{1}{2}\left(B^{\dagger}_i B_i\rho_{\mathcal{S}} +\rho_{\mathcal{S}}B^{\dagger}_i B_i\right)\right],\\
\end{split}
\end{equation}
where $\mathcal{H}_{\mathcal{S}}$ is the Hamiltonian of the subsystem $\mathcal{S}$, $\left\{\gamma_{i};\,i\in I\right\}$ are positive decay rates and $\left\{B_i;\,i\in I\right\}$ are operators acting on $\mathbb{H}_{\mathcal{S}}$ which are usually called \emph{Lindblad operators}, $I$ is a set of indices and $B^{\dagger}_{i}$ denotes the adjoint of $B_{i}$ which in finite dimension is the conjugate transposed of $B_{i}$. The decay rates and the explicit form of the Lindblad operators depend both on the particular global system under study. People usually refers to Eq.(\ref{eq:lindbladdissipator}) as the Lindblad or LGKS master equation.\\
For completeness we also report the dual expression of Eq.(\ref{eq:lindbladdissipator}) which reads
\begin{equation}\label{eq:lindbladdissipatorafjoint}
\begin{split}
L^{*}(O_{\mathcal{S}})=& i\left[\mathcal{H}_{\mathcal{S}},\,O_{\mathcal{S}}\right]+\\
&+\sum_{i\in I}\gamma_{i}\left[B^{\dagger}_i O_{\mathcal{S}} B_i-\frac{1}{2}\left(B^{\dagger}_i B_i O_{\mathcal{S}} + O_{\mathcal{S}} B^{\dagger}_i B_i\right)\right],\\
\end{split}
\end{equation}
being $O_{S}$ $\in$ $B(\mathbb{H}_{\mathcal{S}})$. While Eq.(\ref{eq:lindbladdissipator}) generates the evolution of states in Schr\"{o}dinger representation, Eq.(\ref{eq:lindbladdissipatorafjoint}) describes the time-evolution of operators in Heisenberg picture.\\
\section{A lesson from the simplest open quantum system}\label{sec:review}
After the derivation of the LGKS master equation several theoretical results concerning the general properties of open quantum systems have been published. However, the major part of these papers is written in a formal way and it is often complex to truly undenstand their importance and moreover their applicability. Here, we will consider the most relevant for our purposes. In particular, what we are going to perform is an explicit analysis: we will discuss in details the theoretical assumptions on which these results are based and moreover how they are related to each other. This latter task is done by considering the LGKS equation governing the open dynamics of a two-level system coupled to a bath of harmonic oscillators at zero temperature. As we will show some of them do not apply to this particular problem. Nevertheless, as shown at the end of this section, they can be exploited to characterise efficiently the problem at temperature different from zero.\\
A two-level system coupled to a bath of harmonic oscillators can be described by the following Hamiltonian \cite{petruccionebook}
\begin{equation}\label{eq:twolevelsystemhamiltonian}
\mathcal{H}=\mathcal{H}_{S}+\mathcal{H}_{R}+\mathcal{H}_{I},
\end{equation}
where $\mathcal{H}_{S}$, $\mathcal{H}_{R}$ and $\mathcal{H}_{I}$ denote respectivly the two-level system hamiltonian, the reservoir hamiltonian and the system-reservoir interaction term. While we do not assume any particular form for $\mathcal{H}_{S}$, we do assume the following expressions for the reservoir and system-reservoir hamiltonians:
\begin{equation}\label{eq:twolevelsystem}
\mathcal{H}_{\mathcal{R}}=\sum_{\textbf{k}}\omega_{k}b_{\textbf{k}}^{\dagger}b_{\textbf{k}},
\quad\mathcal{H}_{I}=\sum_{\textbf{k}}\left(g_{\textbf{k}}\sigma^{+}b_{\textbf{k}}+g^{*}_{\textbf{k}}\sigma^{-}b_{\textbf{k}}^{\dagger}\right).
\end{equation}
The $\textbf{k}-$th mode of the reservoir has energy $\omega_{\textbf{k}}$ and is described by a couple of creation and annihilation operators denoted by $b_{\textbf{k}}^{\dagger}$ and $b_{\textbf{k}}$ satisfying standard commutation rules. The two-level system is coupled linearly to the $\textbf{k}-$th mode with a coupling strength $g_{\textbf{k}}$. The two matrices $\sigma^+$ and $\sigma^-$ are the \emph{ladder operators} for the spin-1/2 
\begin{equation}\label{eq:matrices}
\sigma^{+} = \left( \begin{array}{cc}
0 & 1 \\
0 & 0  \end{array} \right),\quad
\sigma^{-} = \left( \begin{array}{cc}
0 & 0 \\
1 & 0  \end{array} \right).
\end{equation}
The Hamiltonian (\ref{eq:twolevelsystemhamiltonian}) can be efficiently exploited to mimick the time-evolution of a two-level atom interacting with the radiation field. In this case, the meaning of the Hamiltonian terms (\ref{eq:twolevelsystemhamiltonian}) is clear: the first two terms describe the free evolution of the atom and the radiation field respectively, while the latter describes absorption and emission processes which involve both matter and photons.\\
Under the Born-Markov approximation and by considering a reservoir at zero temperature $T=0$, the time evolution of the two-level system can be casted in a LGKS fashion:
\begin{equation}\label{eq:LindbladtwolevelsystemDECAY1}
\frac{d}{dt}\rho_{\mathcal{S}}=-i\left[\mathcal{H}^{'}_{S},\rho_{\mathcal{S}}\right]+\gamma\left[\sigma^{-}\rho_{\mathcal{S}}\sigma^{+}-\frac{1}{2}\left\{\sigma^{+}\sigma^{-},\rho_{\mathcal{S}}\right\}\right],
\end{equation}
where $\rho_{\mathcal{S}}$ describes the reduced-density matrix of our system, $\mathcal{H}^{'}_{S}=\mathcal{H}_{S}+\delta \mathcal{H}$ is the two-level Hamiltonian to which we add a Lamb-shift term, $\gamma$ is the decay rate and $\left\{A,\,B\right\}\equiv A B + B A$ denotes the \emph{anticommutator} of the two operators $A$ and $B$.\\
By experience, we know that as the time goes to infinity our atom will end up in the same equilibrium condition \emph{i.e.} there exists a unique steady-state configuration $\rho_{ss}$ satisfying the following equation
\begin{equation}\label{eq:steadystatedef}
L(\rho_{ss})=0.
\end{equation}
Moreover, this happens for any starting condition of our system: no matter what was the atomic state at $t=0$, 
for times much longer than relevant timescales our two-level system will be driven by the radiation field to the same steady-state configuration. However, it is worth noting that this behavior is not \emph{a priori} guaranted for any system. Formally, this happens if the dynamical semigroup governing the open dynamics is \emph{relaxing}. A semigroup $\Lambda_{t}$ is called \emph{relaxing} if there exists a state $\rho_{\infty}\in \mathcal{P}(\mathbb{H})$ such that for every state $\rho\in \mathcal{P}(\mathbb{H})$ 
\begin{equation}
\lim_{t\to +\infty}\Lambda_{t}\rho=\rho_{\infty}.
\end{equation}
While the relaxingness of a dynamical semigroup implies by definition the uniqueness of its steady-state configuration, the converse is not \emph{a priori} true: if there exists a unique solution of the equation (\ref{eq:steadystatedef}), \emph{a priori} it does not imply that any starting configuration will approach it as time goes by. Indeed, if $L$ has pure imaginary eigenvalues, the dynamics allows circular closed paths in $\mathcal{P}(\mathbb{H})$, meaning that there are starting configurations of our system that will not approach to a equilibrium state in the long time limit, but will keep moving along a closed trajectory forever. However, in \cite{baum1} and \cite{schirmer} the authors showed that it is not possible to have such paths. We stress that their result is extremely relevant: it guarantees that the uniqueness of the steady-state is a condition equivalent to relaxingness. As a consequence, uniqueness criteria actually provide not only information about the asymptotic properties of a system, but also on the system dynamics: provided that the the steady-state configuration is unique, any starting configuration will approach it.\\
Let us now analyse the classical papers cited in the introduction, looking for a theoretical result that guarantees what we know about the dynamics of our two-level system. The first we consider here is discussed by Spohn in \cite{spohn1}:
\begin{theorem*}
Let $\mathbb{H}$ be a finite dimensional Hilbert space with $\mbox{dim}(\mathbb{H})=N$ and let the Louvillian superoperator $L$ be given by
\begin{equation}\label{eq:GKSrepresentation}
\begin{split}
L(\rho_{\mathcal{S}})=&-i\,\left[\mathcal{H},\,\rho_{\mathcal{S}}\right]+\\
&+\frac{1}{2}\sum_{i,\,j=1}^{N^2-1}c_{i,j}\left([G_i,\,\rho_{\mathcal{S}}G^{\dagger}_j]+[G_i\rho_{\mathcal{S}},\,G^{\dagger}_j]\right)\\,
\end{split}
\end{equation} 
with $H=H^{\dagger}$, $\mbox{Tr}[G_i]=0$,$\mbox{Tr}[G^{\dagger}_i\,G_j]=\delta_{i,\,j}$ and being ${c_{i,\,j}}$ a complex positive $(N^2-1)\times (N^2-1)$ matrix.\\
Let the positive matrix ${c_{i,\,j}}$ have a p-fold degenerate eigenvalue zero. If $p<N/2$, then the dynamical semigroup is relaxing. In particular, if the matrix ${c_{i,j}}$ is strictly positive, then the dynamical semigroup is relaxing.
\end{theorem*}
This theorem provides a simple sufficient condition for having a relaxing dynamical semigroup and as a consequence a unique steady-state configuration. However, we will show that it does not apply for example to the master equation (\ref{eq:LindbladtwolevelsystemDECAY1}). In order to prove this fact, we need the explicit expression of the $\{c_{i,j}\}$ matrix which can be derived explicitly as discussed below.
The representation given in (\ref{eq:GKSrepresentation}) is equivalent to (\ref{eq:lindbladdissipator}). In terms of the complete set $\{\mathbbm{1},\,G_i\}$ any operator $B_k$ can be decomposed as:\\
\begin{equation}\label{eq:lindbladdecomp}
B_k=\alpha_0^{(k)}\mathbbm{1}+\sum_{j=1}^{N^2-1}\,\alpha^{(k)}_{j}G_j,
\end{equation}
being $\alpha_{0}^{(k)}=\mbox{Tr}[B_{k}]$ and $\alpha_{j}^{(k)}=\mbox{Tr}[B_{k}G_j]$. Plugging (\ref{eq:lindbladdecomp}) in (\ref{eq:lindbladdissipator}), one obtains the representation (\ref{eq:GKSrepresentation}), where the matrix ${c_{i,\,j}}$ is defined by:\\
\begin{equation}
c_{i,\,j}=\sum_{k\in I}\gamma_{k}\alpha^{(k)}_{i}  \alpha^{(k)*}_{j}.  
\end{equation}
Let us apply this result to the master equation (\ref{eq:LindbladtwolevelsystemDECAY1}). In this case, the complete set we can use is given by $\{\mathbbm{1};\,\sigma^{x}/\sqrt{2},\,\sigma^{y}/\sqrt{2},\,\sigma^{z}/\sqrt{2}\}$, being $\sigma^x$,$\sigma^y$ and $\sigma^z$ the Pauli matrices
\begin{equation}
\sigma^{x} = \left( \begin{array}{cc}
0 & 1 \\
1 & 0  \end{array} \right),\quad
\sigma^{y} = \left( \begin{array}{cc}
0 & -i \\
i & 0  \end{array} \right),\quad
\sigma^{z} = \left( \begin{array}{cc}
1 & 0 \\
0 & -1  \end{array} \right).
\end{equation}
In this basis, the Lindblad operator $\sigma^{-}$ has coordinates $(\alpha_0,\alpha_1,\alpha_2,\alpha_3)=(0,1/\sqrt{2},-i/\sqrt{2},0)$ and the corresponding ${c_{i,\,j}}$ matrix reads
\begin{equation}\label{eq:cmatrix}
{c_{i,\,j}}= \left( \begin{array}{ccc}
\gamma/2 & i\gamma/2 & 0\\
-i\gamma/2 &\gamma/2 & 0 \\
0 & 0 & 0 \end{array} \right)
\end{equation}
Since the first and the second row of (\ref{eq:cmatrix}) are proportional, the matrix ${c_{i,\,j}}$ has a 2-fold degenerate eigenvalue zero. This means that in this case the hypothesis required by Spohn cannot be satisfied, making this theoretical result not applicable to this problem ($p>N/2$).\\
In \cite{spohn2} Spohn provided the following sufficient condition for having a relaxing semigroup:
\begin{theorem*} Given a Louvillian superoperator $L$, \emph{if} the linear span of the Lindblad generators entering in $L$, that is $lsp\{B_{i};\,i\in I\}$, is a self-adjoint set and \emph{if} the \emph{bicommutant} of these operators is equal to the set of all the bounded operators $\mathcal{B}(\mathbb{H})$ on the Hilbert space $\mathbb{H}$, that is $\{B_{i};\,i\in I\}''=\mathcal{B}(\mathbb{H})$, then the semigroup is relaxing.
\end{theorem*}
The first hypothesis can be easily checked. The linear span of a set $\{B_{i};\,i\in I\}$ is \emph{self-adjoint} if given any linear combination of the operators $V=\sum_{j}\alpha_jB_j$ in the span, also its adjoint $V^{\dagger}=\sum_{j}\alpha^*_jB^{\dagger}_j$ is an element of the span, being the $\alpha_{j}\in \mathbb{C}$. This constraint is satisfied if (\emph{i}) all the Lindblad operators are self-adjoint \emph{i.e.} $B_{i}=B^{\dagger}_{i}\,\forall\,i\in I$, or if (\emph{ii})  for any $B_{i}$ with $B_{i}\neq B^{\dagger}_{i}$, there is some $j\in I$ such that $B_{j}=B^{\dagger}_{i}$.\\ Satisfying the second hypothesis is more demanding. The \emph{commutant} $\{B_{i};\,i\in I\}'$ of a set $\{B_{i};\,i\in I\}$ is defined as the set of operators simultaneously commuting with all the operators $B_{i}$: $M$ is a element of $\{B_{i};\,i\in I\}'$ if and only if the commutator $[M,\,B_{i}]=0$,$\forall\,i\in\,I$. The \emph{bicommutant} of the set $\{B_{i};\,i\in I\}$, here denoted by $\{B_{i};\,i\in I\}''$, is the commutant of $\{B_{i};\,i\in I\}'$. This means that in order to calculate the bicommutant of a set one needs first to compute its commutant, and then repeat the process on a totally new set of operators.\\
Let us try to apply this result to the master equation (\ref{eq:LindbladtwolevelsystemDECAY1}), forgetting about the second hypothesis provided by Spohn. The results showed below are based on the theory of commutants, so the subject will be analised in details later. The dynamics encoded in the Eq. (\ref{eq:LindbladtwolevelsystemDECAY1}) is characterised by a single Lindblad operator \emph{i.e.} $\sigma^{-}$.
Since $\sigma^-\neq(\sigma^-)^{\dagger}=\sigma^{+}$, we easily see that the first hypothesis disussed above cannot be fulfilled.\\
In \cite{frigerio1} Frigerio provided a sufficient condition for the equivalence between the uniqueness of the stationary state and the triviality of the commutant $\{\mathcal{H},\,B_{i},\,B_{i}^{\dagger};\,i\in I\}'$ \emph{i.e.}  $\{\mathcal{H},\,B_{i},\,B_{i}^{\dagger};\,i\in I\}'=\lambda\,\mathbbm{1}$ . However, as pointed out by Spohn in \cite{spohnrev}, in \cite{frigerio1} one needs to assume by hypothesis that the semigroup has a faithful invariant state \emph{i.e.} a maximum rank steady-state. Therefore, \emph{a priori} this condition restricts considerably the applicability of this theoretical result.\\
In \cite{frigerio2} the author generalised the result given by Spohn in \cite{spohn2}, showing that if the dynamical semigroup $\Lambda_t$ has at least one stationary state $\bar{\rho}$ and if the linear span $lin\{B_{i};\,i\in I\}$ is a self-adjoint set with $\{B_{i};\,i\in I\}'=\lambda\,\mathbbm{1}$, then $\bar{\rho}$ is faithful and we have 
\begin{equation}
\lim_{t\to +\infty}\mbox{Tr}[A\,\Lambda_t\rho]\,=\,\mbox{Tr}[A\bar{\rho}]
\end{equation}
for all $A\,\in\mathcal{B}(\mathbb{H})$ and for all $\rho\,\in\mathcal{P}(\mathbb{H})$.\\
For the same reasons discussed above, since $\sigma^{-}$ is not self-adjoint, the hypothesis required in \cite{frigerio2} cannot be fulfilled by the master equation (\ref{eq:LindbladtwolevelsystemDECAY1}).\\
The last result we discuss here is \cite{evans}. Therein, Evans provides in a extremely formal way a necessary and sufficient condition for the irreducibility of a dynamical semigroup:
\begin{theorem*}
Let $\mathcal{M}$ be a von Neumann algebra on the Hilbert space $\mathbb{H}$, which is globally invariant under the semigroup $\Lambda_{t}=e^{tL}$, being the Louvillian superoperator $L$ given by
\begin{equation}\label{eq:lindbladEvans}
L(X)=V(X)+K^{\dagger}\,X+X\,K,\quad \mbox{with}\,X\in\mathcal{B}(\mathbb{H}),
\end{equation} 
and being
\begin{equation}\label{eq:Vequation}
V(X)=\int_{\Omega}A(\omega)^{\dagger}XA(\omega)\mbox{d}\mu(\omega),\quad K=i\mathcal{H}-\frac{1}{2}V(\mathbbm{1}).
\end{equation}
If $T_t$ is the induced W*-dynamical semigroup on M, then the algebra $\mathcal{M}(T)=\{X\in M:T_{t}(X^{\dagger}X)=X^{\dagger}X,\,T_{t}(X)=X,\,\forall t\geq 0\}$ is equal to $\mathcal{M}\cap \{A(\omega),\,K\}'$. Thus $T_t$ is irreducible if and only if $\mathcal{M}\cap \{A(\omega),\,A(\omega)^{\dagger},\,\mathcal{H}\}'=\mathbb{C}\mathbbm{1}$.
\end{theorem*}
Before discussing this result, let us observe that the representation given in (\ref{eq:lindbladEvans}) is equivalent to the standard representation of the Louvillian. This can be easily shown by substituing the formal definition of $V(X)$ with the following
\begin{equation}\label{eq:Vsubstitution}
V(X)=\int_{\Omega}A(\omega)^{\dagger}X A(\omega)\mbox{d}\mu(\omega)\to V(X)=\sum_{i\in I}A_{i}^{\dagger}X A_i.
\end{equation}
This result by Evans may look quite similar to those listed above, however there are few aspects which make it more applicable \emph{a priori}: it does not assume any particular structure for the Lindblad operators or their linear span, it provides a necessary and sufficient condition and it does not assume any particular form of the steady-state solution. In addition, it applies successfully to the simple problem described by the master equation (\ref{eq:LindbladtwolevelsystemDECAY1}). Indeed, given the set $\{\sigma^{-},\,\sigma^{+},\,\mathcal{H}\}$, the commutant $\{\sigma^{-},\,\sigma^{+},\,\mathcal{H}\}'$ contains only operators proportional to the identity $\mathbbm{1}$. This can be easily seen: since any operator $X$ which commutes with both $\sigma^{-}$ and $\sigma^{+}$ also commutes with $\left[\sigma^{+},\,\sigma^{-}\right]\propto \sigma^{z}$, by Schur's Lemma $X=\lambda\,\mathbbm{1}$, because $\left\{\sigma^{+},\,\sigma^{-},\,\sigma^z/2\right\}$ is a $2$-dimensional irreducible representation of the Lie algebra of the $SU(2)$ group \cite{georgi}. Since any operator in the commutant a of a given set must commute with all the elements of the set, all the elements of $\{\sigma^{-},\,\sigma^{+},\,\mathcal{H}\}'$ must be proportional to the identity.\\
However, in order to appreciate the physical consequences and the meaning of this result, we need to spend a few words about what \emph{irreducibility} means. In a recent paper, Baumgartner and Narnhofer \cite{baum2} pointed out that there is a relation between the possibility of decomposing a Hilbert space into invariant subspaces and the presence of dynamical symmetries (see also \emph{e.g.} \cite{albert}). The existence of invariant subspaces means that the Lindblad equation conserves a set of orthogonal projectors $\{P_{i}\}$, with $P_i=P_i^{\dagger}=P_i^2$. This happens if all the operators $\{P_i\}$ commute simultaneously with the Hamiltonian $\mathcal{H}$, with all the dissipators $B_i$ and with their adjoint $B^{\dagger}_i$, \emph{i.e.} the projectors are elements of the commutant $\{B_{i},\,B^{\dagger}_i,\,\mathcal{H}\}'$. Let us suppose to have a LGKS conserving two orthogonal projectors $P_1$ and $P_2$ with orthogonal support \emph{i.e.} projectors on orthogonal subspaces $W_1$ and $W_2$ of the Hilbert space. In this case, our open quantum system has more than a single steady-state configuration. Indeed, since $P_1$ ($P_2$) is conserved by the time-evolution, the semigroup maps density matrices representing states in $W_1$ ($W_2$) into convex superpositions of states in $W_1$ ($W_2$). Therefore, in the large time limit any state starting in $W_1$ ($W_2$) converges to a steady-state configuration in $W_1$ ($W_2$).
When the only projector conserved by the time-evolution is the identity (trivial projection operation), it is not possible to find a proper subspace of the total Hilbert space which is invariant under the action of the semigroup. In this case, the semigroup is \emph{irreducible}, the commutant $\{B_{i},\,B^{\dagger}_i,\,\mathcal{H}\}'$ contains only operators proportional to the identity and most importantly the steady-state is unique.\\
It is worth noting that a different definition of irreducibility was given by Davies in \cite{daviesstochastic}, while discussing the properties of quantum stochastic processes. In \cite{daviesstochastic} the author has shown results which implies that a dynamical semigroup on a Hilbert space $\mathbb{H}$ is irreducible if and only if there is no proper closed subspace $\mathbb{K}\subset \mathbb{H}$ invariant under the dynamics (see also \emph{e.g.} \cite{schirmer}). As observed by the same Evans in \cite{evans}, Spohn in \cite{spohn2} and by Frigerio in \cite{frigerio1}, these two definitions are strictly related. Indeed, while Evans derived his results by considering the action of the semigroup on operators \emph{i.e.} working in Heisenberg picture, Davies derived his results by working in Schr\"{o}dinger picture \emph{i.e.} by considering the effects of the open quantum dynamics directly on the Hilbert space.\\ 
In addition, it is worth noting at this point that the scenario would have been different for $T\neq 0$. For $T\neq 0$, the master equation  governing our simple two-level system reads
\begin{equation}\label{eq:LindbladtwolevelsystemDECAY2}
\begin{split}
\frac{d}{dt}\rho_{\mathcal{S}}=&-i\left[\mathcal{H}'_{S},\rho_{\mathcal{S}}\right]+\gamma_{\downarrow}\left[\sigma^{-}\rho_{\mathcal{S}}\sigma^{+}-\frac{1}{2}\{\sigma^{+}\sigma^{-},\rho_{\mathcal{S}}\}\right]+\\
&+\gamma_{\uparrow}\left[\sigma^{+}\rho_{\mathcal{S}}\sigma^{-}-\frac{1}{2}\{\sigma^{-}\sigma^{+}\rho_{\mathcal{S}}\}\right],
\end{split}
\end{equation} 
where $\gamma_{\downarrow}=\gamma(1+N(\bar{\omega}))$, $\gamma_{\uparrow}=\gamma N(\bar{\omega})$ and $N(\bar{\omega})$ is the mean occupation of the harmonic bath at temperature $T$ in correspondence of a energy gap between the two-levels of $\bar{\omega}$. In this case, since the two Lindblad operators $\sigma^+$ and $\sigma^-$ are related by  $(\sigma^-)^{\dagger}=\sigma^{+}$ and since their commutant is trivial, we easily see that the result derived by Spohn in \cite{spohn2} and by Frigerio in \cite{frigerio2} guarantee that the dynamical semigroup is relaxing and that the unique steady-state configuration is faithful.\\

\section{Theoretical results $\&$ Applications}\label{sec:theorems}
In the previous section we discussed the most relevant and most cited papers concerning the determination of criteria that guarantee the uniqueness of the steady-state solution of the LGKS equation, the relaxingness of a dynamical semigroup and its irreducibility. As shown, the uniqueness of the steady-state is a necessary and sufficient condition for the relaxingness. The same is true for irreducibility. In particular, we have shown that among all the classical papers concerning these subjects the result provided by Evans \cite{evans} is the only capable of predicting the behavior of the simplest open quantum system \emph{i.e.} a two-level system coupled to a Markovian bath of harmonic oscillators. The success of this result relies on the particular structure of the Lindblad operators entering in the master equation (\ref{eq:LindbladtwolevelsystemDECAY1}).
Inspired by this fact we have derived a series of practical sufficient criteria which guarantee the irreducibility of the dynamical semigroups governing single and composite open quantum systems. In addition, we stress that as a consequence of what discussed above, our results guarantee also the uniqueness of the steady-state solution and its attractivity \emph{i.e.} any starting condition approaches to the same fixed point as the time goes to infinity.\\ 
A complete and rigorous proof of the theorems shown below is given in the Appendix.\\

\subsection{Single open quantum systems}

Let us consider a single quantum system described by a Hilbert space $\mathbb{H}$ having finite dimension $d$ \emph{i.e.} $\mbox{dim}\left(\mathbb{H}\right)=d$,
and let us consider the following two operators
\begin{equation}\label{eq:ladderoperatorform}
\bar{B}=\sum_{k=1}^{d-1} (\bar{B})_{k+1,\,k} \vert k+1 \rangle \langle k \vert ,\quad (\bar{B})_{k+1,\,k}\neq \, 0\, \forall\, k,
\end{equation}
\begin{equation}\label{eq:ladderoperatorformpiu}
\bar{B}^{\dagger}=\sum_{k=1}^{d-1} (\bar{B}^{\dagger})_{k,\,k+1} \vert k \rangle \langle k+1 \vert,\quad(\bar{B}^{\dagger})_{k,\,k+1}\neq \, 0\, \forall\, k 
\end{equation} 
where $\{\vert k \rangle \}$ represents a basis of the Hilbert space $\mathbb{H}$ in the standard Dirac notation.  
We are now in the position of stating our first result:
\begin{theorem}\label{th:primo} 
Given a system described by a Hilbert space $\mathbb{H}$ with finite dimension $d$, whose open quantum dynamics is governed by a dynamical semigroup with generator given in the LGKS form, if there exists a linear complex combination
\begin{equation}
M=\alpha_0 \mathcal{H}_{\mathcal{S}}+\sum_{i\,\in I}\left(\alpha_i B_i+\beta_i  B^{\dagger}_i \right)
\end{equation}
and a unitary operator $U$ such that 
\begin{equation}
\bar{B}=U^{\dagger}MU\quad \mbox{or}\quad\bar{B}^{\dagger}=U^{\dagger}MU,
\end{equation}
then the dynamical semigroup $\Lambda_t=exp(Lt)$ is irreducible.
\end{theorem}
Though at first glance, this result that may look similar to those listed above, it actually provides an easy and extremely applicable criterion which successfully describes relevant problems in physics. Indeed, as we show in the following, a large class of open quantum systems satisfies the constraints required by our theorem.\\
\newline
As a first example, let us consider a $N$-level atom coupled to the radiation field at temperature $T\neq0$. The system is schematically represented in Fig. \ref{fig:Nlevelscheme}.\\ 
\begin{figure}[H]
\centering
\begin{tikzpicture}[scale=0.75]

\draw [->, thick] (-1.2,-4.2)--(-1.2,0.2);
\node at (-0.9,0.15) {\large $E$};

\draw [thick] (0,0)--(2,0);
\node  at (2.55,0) {$\vert N \rangle$};
\draw [thick] (0,-1)--(2,-1);
\draw [<->,dashed,very thick] (1,-1.1)--(1,-1.9);
\node  at (2.85,-1) {$\vert N -1 \rangle$};
\draw [thick] (0,-2)--(2,-2);
\node  at (2.5,-2) {$\vert 3 \rangle$};
\draw [thick] (0,-3)--(2,-3);
\node  at (2.5,-3) {$\vert 2 \rangle$};
\draw [thick] (0,-4)--(2,-4);
\node at (2.5,-4) {$\vert 1 \rangle$};
\draw [red,thick] (0.8,-0.9) to [out=120,in=-90] (0.7,-0.5);
\draw [red,thick,->] (0.7,-0.5) to [out=90,in=-120] (0.8,-0.1);
\draw [red,thick] (0.8,-2.9) to [out=120,in=-90] (0.7,-2.5);
\draw [red,thick,->] (0.7,-2.5) to [out=90,in=-120] (0.8,-2.1);
\draw [red,thick] (0.8,-3.9) to [out=120,in=-90] (0.7,-3.5);
\draw [red,thick,->] (0.7,-3.5) to [out=90,in=-120] (0.8,-3.1);

\draw [blue,thick,<-] (1.2,-0.9) to [out=60,in=-90] (1.3,-0.5);
\draw [blue,thick] (1.3,-0.5) to [out=90,in=-60] (1.2,-0.1);
\draw [blue,thick,<-] (1.2,-2.9) to [out=60,in=-90] (1.3,-2.5);
\draw [blue,thick] (1.3,-2.5) to [out=90,in=-60] (1.2,-2.1);
\draw [blue,thick,<-] (1.2,-3.9) to [out=60,in=-90] (1.3,-3.5);
\draw [blue,thick] (1.3,-3.5) to [out=90,in=-60] (1.2,-3.1);
\end{tikzpicture}
\caption{Representation of a N-level atom interacting with the radiation field at non zero temperature $T\neq 0$. Levels close in energy (E) interact exchanging photons with the radiation field: absorption and emission processes are represented respectively by arrows going upwards and downwards.}\label{fig:Nlevelscheme}
\end{figure}
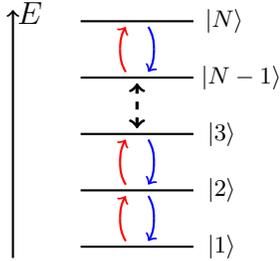 
Interactions with the radiatin field, as depicted in Fig.\ref{fig:Nlevelscheme}, drive a $N$-level atom into a steady state configuration which does not depend on the initial configuration of our open quantum system. 
This means that the underlying time-evolution is governed by an irreducible dynamical semigroup. This can be shown using Theorem \ref{th:primo}. Emission and absorption processes are represented respectively by Lindblad operators having the following structure
\begin{equation}\label{eq:NlevelDissipators}
B_i= \vert i \rangle \langle i+1 \vert,\quad i\in [1,\cdots,N-1],
\end{equation}
\begin{equation}\label{eq:NlevelDissipatorsPiu}
B^{\dagger}_i= \vert i+1 \rangle \langle i \vert,\quad i\in [1,\cdots,N-1],
\end{equation}
where $\{\vert i \rangle \}$ denote the $N$ levels of our atom.\\
It is easy to see that provided the set of operators in (\ref{eq:NlevelDissipators}) and in (\ref{eq:NlevelDissipatorsPiu}), the linear combination $M=\sum_{i=1}^{N-1}B_i$ satisfies the hypothesis of our theorem, proving that the dynamics is governed by an irreducible dynamical semigroup.\\
Moreover, our result applies directly to those systems where at least one among the Lindblad operators $B_i$ has the structure (\ref{eq:ladderoperatorform}) or (\ref{eq:ladderoperatorformpiu}): 
\begin{corollary}\label{cor:coruno}
If at least one among the Lindblad operators entering in the LGKS equation has a structure unitarily equivalent to that given in (\ref{eq:ladderoperatorform}) or (\ref{eq:ladderoperatorformpiu}), then the dynamical semigroup is irreducible.
\end{corollary} 
For example, this happens in spin systems or in boson systems with truncated Fock space, if there is a decay channel driven respectively by a ladder operator or by a annihilation or annihilation operator.
Indeed, the structure of $S^{-}$
\begin{equation}
S^{-}=\sum_{m=-S+1}^{S}\sqrt{(S+m)(S-m+1)}\,\vert\left.m-1\right\rangle\left\langle\,m\right.\vert,
\end{equation}
and the structure of the annihilation operator $A$
\begin{equation}
A=\sum_{n=1}^{N_{max}}\,\sqrt{n} \,\vert \left. n-1\right\rangle\left\langle\, n \right.\vert,
\end{equation}
are exactly as the one required by our theorem, 
where $S$ is the spin value, $N_{max}$ is the number of states in the Fock basis, and $\{\vert m \rangle\}$ and $\{\vert n\rangle\}$ denote respectively the eigenstates of the z-component of the spin $S_z$ and the number operator $A^{\dagger}A$.\\
It is worth noting that the truncation of the Fock space in boson systems is a standard procedure in numerical simulations. Indeed, most of the numerical results on these systems are obtained by checking how results scale while increasing the Fock space dimension. Our results guarantee that a numerical search for the steady-state solution always leads to a unique solution, independently by the number of states kept to reach the desired convergence.\\
In the following section we generalise these results to the case of composite quantum systems.\\ 
\subsection{Composite open quantum systems} 
Let us consider a physical system described by a \emph{total} Hilbert space $\mathbb{H}$ that is the tensor product of a finite number $N$ of Hilbert spaces $\mathbb{H}_i$ having finite dimension $d_i$, that is $\mathbb{H}=\bigotimes_{i=1}^{N}\mathbb{H}_{i}$, being $d_i=\mbox{dim}(\mathbb{H}_{i})$. 
Let us consider the following set of operators $\{\bar{B}^{(N)}_i;\,i=1,\cdots,\,N\}$ 
\begin{equation}\label{eq:B1}
\bar{B}^{(N)}_{1}=\bar{B}_1\,\otimes\mathbbm{1}_{d_2} \otimes \mathbbm{1}_{d_{3}} \otimes \cdots \otimes \mathbbm{1}_{d_N},
\end{equation}
\begin{equation}\label{eq:Bj}
\bar{B}^{(N)}_{j}=1_{d_1}\otimes\cdots\otimes \mathbbm{1}_{d_{j-1}}\otimes\,\bar{B}_j\,\otimes\mathbbm{1}_{d_{j+1}}\otimes\cdots\otimes \mathbbm{1}_{d_N},\,\,(1<j<N)
\end{equation}
\begin{equation}\label{eq:Bn}
\bar{B}^{(N)}_{N}=\mathbbm{1}_{d_1} \otimes \mathbbm{1}_{d_2} \otimes \cdots \otimes \mathbbm{1}_{d_{N-1}}\otimes \bar{B}_N\,,
\end{equation}
where $\{\bar{B}_i\}$ are an operators having the structure in (\ref{eq:ladderoperatorform}) and $\mathbbm{1}_{d_i}$ represents the $d_i$-dimensional identity. The symbol ``$\otimes$" in Eqs.(\ref{eq:B1})-(\ref{eq:Bn}) denotes the \emph{Kronecker product} between operators.
\begin{theorem}\label{th:secondo} 
Given a composite system described by a total Hilbert space $\mathbb{H}$ made up of a finite number of finite dimensional constituents \emph{i.e.} $\mathbb{H}=\bigotimes_{i=1}^{N}\mathbb{H}_{i}$ and $\mbox{dim}(\mathbb{H}_{i})=d_i < \infty$, whose open quantum dynamics is governed by a dynamical semigroup with generator given in the LGKS form, if there exists a set of linear complex combinations $\{M_j;\,j\in[1,\cdots,N]\,\}$
\begin{equation}
M_j=\alpha_{0,\,j} \mathcal{H}_{\mathcal{S}}+\sum_{i\,\in I}\left(\alpha_{i,\,j} B_i+\beta_{i,\,j}  B^{\dagger}_i \right)
\end{equation}
and a unitary operator $U$ such that 
\begin{equation}
\bar{B}^{(N)}_j=U^{\dagger}M_jU\quad \mbox{or}\quad\bar{B}^{(N)\,\dagger}_{j}=U^{\dagger}M_jU,
\end{equation}
then the dynamical semigroup $\Lambda_t=exp(Lt)$ is irreducible, being the $\{\bar{B}^{(N)}_{j}\}$ specified in eqs.(\ref{eq:B1})-(\ref{eq:Bn}) and being the $\{\bar{B}^{(N)\dagger}_{j}\}$ their adjoint.\\
\end{theorem}
In addition, we have the following corollary
\begin{corollary}\label{cor:cordue}
If at least $N$ among the Lindblad operators entering in the LGKS equation are unitarily equivalent to the operators $\{\bar{B}^{(N)}_{i}\}$ or to their adjoint $\{\bar{B}^{(N)\dagger}_{i}\}$, then the dynamical semigroup is irreducible.
\end{corollary}
These results provide simple sufficient conditions for composite quantum systems which can be exploited in many different cases. Theorem \ref{th:secondo} applies to the master equation describing the open time-evolution of a collection of $N$ $M$-level atoms coupled to the radiation field. If one considers each atom as coupled to a bath, the Lindblad operator inducing the decay of the excited state $\vert i+1 \rangle$ of the $j$-th atom to the state $\vert i\rangle$ has the following structure
\begin{equation}\label{eq:Nleveldecaymanybody}
B_{i,\,j}= \underbrace{\mathbbm{1}_{M}\otimes \cdots\mathbbm{1}_{M}}_{j-1\,\mbox{times}}\otimes \vert i \rangle \langle i+1 \vert\otimes\underbrace{\mathbbm{1}_{M}\otimes\cdots\otimes \mathbbm{1}_{M}}_{N-j\,\mbox{times}}
\end{equation}
and it is easy to see that the following linear combination 
\begin{equation}
\begin{split}
M_j&=\sum_{i=1}^{N-1}B_{i,\,j}=\\
&=\mathbbm{1}_{M}\otimes \cdots\mathbbm{1}_{M}\otimes \sum_{i=1}^{N-1} \vert i \rangle \langle i+1 \vert\otimes\mathbbm{1}_{M}\otimes\cdots\otimes \mathbbm{1}_{M}.\\
\end{split}
\end{equation}
has the structure specified in (\ref{eq:Bj}), meaning that the underlying dynamics is governed by an irreducible semigroup.\\
Theorem \ref{th:secondo} and Corollary \ref{cor:cordue} find a direct application to the physics of open quantum lattices of bosons and spins. As pointed out in the introduction, these systems have been intensively studied during the last decade (see \emph{e.g.} \cite{lee,rota,jin2,jin,lebreuilly}). 
In addition, we have to stress that theorem \ref{th:secondo} applies also to cases where the constituents do have different properties \emph{i.e.} they represent different physical subsystems. This is for example the case of \cite{biella} and \cite{restrepo}. In the first paper the authors consider an array of coupled cavities in presence of incoherent driving and dissipation, where each site is a Kerr nonlinear resonator coupled to a two-level emitter pumped incoherently. In the latter, the effects of driving and dissipation terms on a fully coupled hybrid optomechanical system where all mutual couplings between a two-level atom, a confined photon mode and a mechanical oscillator mode have been considered. In both cases, the dissipation channels are induced by ladder and annihilation operators. Corollary \ref{cor:cordue} guarantees that the steady-state configuration is unique.\\
\section{Summary and Conclusions}\label{sec:conclusions}
We have briefly reviewed the most relevant theoretical results developed during the last forty years concerning the uniqueness of the steady-state solution of the LGKS master equation and the determination of conditions under which a dynamical semigroup is relaxing or irreducible. In particular, we have discussed in details how these results apply to the simplest open quantum system \emph{i.e.} a two-level system coupled to a reservoir of harmonic oscillators at zero temperature, showing that only the result provided by Evans \cite{evans} is capable of justifying the well-known behavior of this prototypical model. In doing so, we have shown also how these issues are related. The uniqueness of the steady-state is a necessary and sufficient condition for having a relaxing dynamical semigroup. The same is true for irreducibility.\\
In addition, we have provided a set of sufficient conditions (rigorously proved in the Appendix) which can be exploited to characterise the asymptotic properties of quantum dynamical semigroups. By considering the linear combinations of the Lindblad operators in the LGKS equation and by using simple results from the theory of operator algebras, we have shown that it is possible to determine if a dynamical semigroup is irreducible \emph{i.e.} the steady-state solution is unique and attractive. Theorem \ref{th:primo} and Corollary \ref{cor:coruno} provide simple criteria for single quantum systems. These results have been generalised in Theorem \ref{th:secondo} and Corollary \ref{cor:cordue} to address also the properties of composite open quantum systems. Our results find a direct application in the characterisation of a wide variety of problems that recently have been under investigation, such as $d$-dimensional lattice systems of spins and bosons, and heterogeneous systems such as hybrid optomechanical systems.\\

\appendix*
\section{Proofs of Theorems in \ref{sec:theorems}}\label{app:proofs}
In this section we give an explicit proof of the two theorems discussed in section \ref{sec:theorems}. 
The key point of our discussion is the following. Since by Evans results \cite{evans} a dynamical semigroup governing the time-evolution is \emph{irreducible} if and only if $\{\mathcal{H}_{\mathcal{S}},\,B_{i},B^{\dagger}_{i};\, i\in I\}'=\lambda \mathbbm{1}$, if one is able to prove the triviality of the commutant, then the uniqueness of the steady-state follows. This task is considerably simplified if one is able to find a subset $\{\bar{B},\,\bar{B}^{\dagger}\}$ of $\{\mathcal{H}_{\mathcal{S}},\,B_{i},B^{\dagger}_{i};\, i\in I\}$ whose commutant is trivial. Indeed, given a set $\Omega$ and $\bar{\Omega}\subseteq \Omega$, the commutant of $\Omega$ is contained in $\bar{\Omega}$ \emph{i.e.} $\Omega ' \subseteq \bar{\Omega}'$ \cite{black}. This means that if $\bar{\Omega}'=\lambda \mathbbm{1}$, since the commutant of any set contains the identity and its multiples and $\lambda \mathbbm{1}=\bar{\Omega}'\supseteq \Omega '$, we have that $\Omega '= \lambda \mathbbm{1}$. By applying this idea to our problem we derived the theorems discussed in the main text.\\
Since in the following we use this structure many times, let us rephrase the operators $\bar{B}$ and $\bar{B}^{\dagger}$ as follows:
\begin{equation}\label{eq:ladderoperatorformapp}
\bar{B}=\sum_{k=1}^{d-1} (\bar{B})_{k+1,\,k} E_{k+1,\,k},\quad (\bar{B})_{k+1,\,k}\neq \, 0\, \forall k\in\left[1,\,d-1\right],
\end{equation}
\begin{equation}\label{eq:ladderoperatorformpiuapp}
\bar{B}^{\dagger}=\sum_{k=2}^{d} (\bar{B}^{\dagger})_{k-1,\,k} E_{k-1,\,k},\quad(\bar{B}^{\dagger})_{k-1,\,k}\neq \, 0\, \forall k\in\left[2,\,d\right],
\end{equation}
where $E_{i,\,j}$ denotes the $d \times d$ square matrix, whose elements are all zero except for the element in the $i$-th row and $j$-th column whose value is 1.\\
The first step for proving our theorems is given by the following proposition:
\begin{prop}\label{prop:uno}
Given the operators $\bar{B}$ and its adjoint $\bar{B}^{\dagger}$ in the form (\ref{eq:ladderoperatorformapp}) and (\ref{eq:ladderoperatorformpiuapp}), the commutant $\left\{\bar{B},\,\bar{B}^{\dagger}\right\}'$ is trivial, that is\\
\begin{equation}
\left\{\bar{B},\,\bar{B}^{\dagger}\right\}'=\lambda\, \mathbbm{1}_{d}\,
\end{equation}
\end{prop}

\begin{proof}
Let us consider a generic operator $X\,\in\mathcal{B}(\mathbb{H})$. It can be decomposed in terms of the set $\left\{E_{i,\,j}\right\}$ in the following way:
\begin{equation}\label{eq:generalmatrixX}
X=\sum_{i,\,j=1}^{d} (X)_{i,\,j}\,E_{i,\,j}
\end{equation}
By using the two representations in Eq.(\ref{eq:ladderoperatorformpiuapp}) and Eq.(\ref{eq:ladderoperatorformapp}) and the commutation rule for the matrices $\left\{E_{i,\,j}\right\}$, that is
\begin{equation}\label{eq:commutatorEiEj}
\left[E_{i,\,j},\,E_{w,\,k}\right]=\delta_{j,\,w}\, E_{i,\,k} - \delta_{k,\,i}\, E_{w,\,j},
\end{equation}
we obtain that 
\begin{equation}\label{eq:commpiu}
\begin{split}
\left[X,\,\bar{B}^{\dagger}\right]=&\sum_{l=1}^{d}\sum_{i=2}^{d}(X)_{l,\,i-1}\,(\bar{B}^{\dagger})_{i-1,\,i}\,E_{l,\,i}-\\
&-\sum_{l=1}^{d-1}\sum_{i=1}^{d}(X)_{l+1,\,i}\,(\bar{B}^{\dagger})_{l,\,l+1}\,E_{l,\,i},\\
\end{split}
\end{equation}
and
\begin{equation}\label{eq:commmeno}
\begin{split}
\left[X,\,\bar{B}\right]=&\sum_{l=1}^{d}\sum_{i=1}^{d-1}(X)_{l,\,i+1}\,(\bar{B})_{i+1,\,i}\,E_{l,\,i}-\\
&-\sum_{l=2}^{d}\sum_{i=1}^{d}(X)_{l-1,\,i}\,(\bar{B})_{l,\,l-1}\,E_{l,\,i}.\\
\end{split}
\end{equation}
The operator $X$ belongs to the commutant $\left\{\bar{B},\,\bar{B}^{\dagger}\right\}'$ if and only if the two matrices in Eq.(\ref{eq:commpiu}) and in Eq.(\ref{eq:commmeno}) are both equal to $0_{d}$, that is the $d\times d$  matrix with all entries equal to zero. Since the matrices $E_{l,\,i}$ are all linearly independent, this happens if and only if the coefficient multiplying each matrix $E_{l,\,i}$ is 0.\\
Since in Eq.(\ref{eq:commpiu}) the first summation does not depend by $i=1$ and by hypothesis all the $(\bar{B}^{\dagger})_{i-1,\,i}\neq 0$ (see Eq.(\ref{eq:ladderoperatorformpiu})), we have that $(X)_{l+1,\,1}=0,\,\forall\,l\in [1,\,d-1]$ (all the elements in the first column of the matrix $X$ except $(X)_{1,\,1}$ are zero). In the same way, by considering Eq.(\ref{eq:commmeno}), since the summation in the second term does not depend on $l=1$, all the coefficients for $l=1$ in the first term must be equal to zero. This means that $(X)_{1,\,i+1}=0,\,\forall\,i\in[1,\,d-1]$ (all the elements in the first row of the matrix $X$ except $(X)_{1,\,1}$ are zero). \\
Let us now consider in Eq.(\ref{eq:commpiu}) the coefficients for $i=2$. In this case, for $l=1$ we obtain
\begin{equation}
(X)_{1,\,1}(\bar{B}^{\dagger})_{1,\,2}-(X)_{2,\,2}(\bar{B}^{\dagger})_{1,\,2}=0\,\Rightarrow\, (X)_{1,\,1}=(X)_{2,\,2},
\end{equation}
while for $l\,>\,1$, since $(X)_{l+1,\,1}=0,\,\forall\,l\in [1,\,d-1]$, we have that $(X)_{l+1,\,2}=0,\,\forall\,l\in [2,\,d-1]$. In the same way, by considering $l=2$ in Eq.(\ref{eq:commmeno}), we obtain again that $(X)_{1,\,1}=(X)_{2,\,2}$ from $i=2$ and since $(X)_{1,\,i+1}=0,\,\forall\,i\in[1,\,d-1]$ we obtain $(X)_{l,\,i+1}=0,\,\forall\,i\in[2,\,d-1]$. This means that after the steps described above the only entries of the matrix $X$ in the first two columns and rows which could be non-zero are the diagonal entries $(X)_{1,\,1}$ and $(X)_{2,\,2}$. In addition, $(X)_{1,\,1}=(X)_{2,\,2}$. If we now consider the terms for $i=3$ in Eq.(\ref{eq:commpiu}) and the terms for $l=3$ in Eq.(\ref{eq:commmeno}), we obtain that $(X)_{2,\,2}=(X)_{3,\,3}$ and that $(X)_{l+1,\,3}=0,\,\forall\,l\in[3,\,d-1]$ and $(X)_{3,\,i+1}=0,\,\forall\,i\in[3,\,d-1]$.\\
In this way at the step $k$ ( $k>2 )$, from Eq.(\ref{eq:commpiu}) by setting $i=k$ one obtains
\begin{equation}
\begin{split}
\left(\left[X,\,\bar{B}^{\dagger}\right]\right)^{k}=&\sum_{l=1}^{d}(X)_{l,\,k-1}\,(\bar{B}^{\dagger})_{k-1,\,k}\,E_{l,\,k}-\\
&-\sum_{l=1}^{d-1}(X)_{l+1,\,k}\,(\bar{B}^{\dagger})_{l,\,l+1}\,E_{l,\,k},
\end{split}
\end{equation}
while from Eq.(\ref{eq:commmeno}) by setting $l=k$ one gets
\begin{equation}
\begin{split}
\left(\left[X,\,\bar{B}\right]\right)_{k}=&\sum_{i=1}^{d-1}(X)_{k,\,i+1}\,(\bar{B})_{i+1,\,i}\,E_{k,\,i}-\\
&-\sum_{i=1}^{d}(X)_{k-1,\,i}\,(\bar{B})_{k,\,k-1}\,E_{k,\,i},
\end{split}
\end{equation}
where $(O)^{k}$ and $(O)_{k}$ denote respectively the $k$-th column and the $k$-th row of the matrix $O$.\\
Thanks to the constraint obtained by the previous $k-1$ steps, these two equations reduce to
\begin{equation}
(X)_{k-1,\,k-1}=(X)_{k,\,k}
\end{equation}
and to 
\begin{equation}
\begin{split}
&(X)_{l+1,\,k}=0,\,\forall\,l\in[k,\,d-1]\,,\\
&(X)_{k,\,i+1}=0,\,\forall\,i\in[k,\,d-1].
\end{split}
\end{equation}
In this way, after a finite number of steps, we obtain that the matrix $X$ is diagonal and proportional to $\mathbbm{1}_d$. For the arbitrariness in the choice of $X$, we conclude that any matrix which commutes simultaneously with both $\bar{B}$ and $\bar{B}^{\dagger}$ \emph{must} be of the form $X=\lambda\,\mathbbm{1}_{d}$, with $\lambda\in\mathbb{C}$. Therefore, we have that
\begin{equation}
\left\{\bar{B},\,\bar{B}^{\dagger}\right\}'=\lambda\, \mathbbm{1}_{d}.
\end{equation}
\end{proof} 
We are now in the position of proving Theorem \ref{th:primo}.\\
\newline
\textbf{Theorem 1.} 
{\em Given a system described by a Hilbert space $\mathbb{H}$ with finite dimension $d$, whose open quantum dynamics is governed by a dynamical semigroup with generator given in the LGKS form, 
if there exists a linear complex combination}
\begin{equation}
M=\alpha_0 \mathcal{H}_{\mathcal{S}}+\sum_{i\,\in I}\left(\alpha_i B_i+\beta_i  B^{\dagger}_i \right)
\end{equation}
{\em and a unitary operator $U$ such that} 
\begin{equation}
\bar{B}=U^{\dagger}MU\quad \mbox{or}\quad\bar{B}^{\dagger}=U^{\dagger}MU,
\end{equation}
{\em then the dynamical semigroup $\Lambda_t=exp(Lt)$ is irreducible.}\newline
\begin{proof}
Let us consider an operator $X$ belonging to the commutant $\{\mathcal{H}_{\mathcal{S}},\,B_{i},B^{\dagger}_{i};\, i\in I\}'$. By definition we have
\begin{equation}
[X,\,B_i]=[X,\,B^{\dagger}_i]=[X,\,\mathcal{H}_{\mathcal{S}}]=0,
\end{equation}
and as a consequence we have also
\begin{equation}\label{eq:commutatorswithM}
[X,\,M]=[X,\,M^{\dagger}]=0.
\end{equation}
By hypothesis there exists a unitary operator $U$ such that $\bar{B}=UMU^{\dagger}$ or $\bar{B}^{\dagger}=UMU^{\dagger}$. Since in the following we use both $\bar{B}$ and $\bar{B}^{\dagger}$, let us suppose to have
\begin{equation}
\bar{B}=UMU^{\dagger},\quad \bar{B}^{\dagger}=UM^{\dagger}U^{\dagger} 
\end{equation}
and let us define $X^{U}=UXU^{\dagger}$. The proof follows the same steps in the other case. By using (\ref{eq:commutatorswithM}), it follows that
\begin{equation}\label{eq:commutatorwithbarbi}
[X^U,\,\bar{B}]=[X^U,\,\bar{B}^{\dagger}]=0.
\end{equation} 
However, we have shown in Proposition \ref{prop:uno} that the condition reported in (\ref{eq:commutatorwithbarbi}) is equivalent to $X^U=\lambda\mathbbm{1}$. This implies that 
\begin{equation}\label{eq:Xdiagonal}
X=U^{\dagger}X^{U}U=\lambda \mathbbm{1}.
\end{equation}
For the arbitrariness in the choice of $X$, we conclude that $\{\mathcal{H}_{\mathcal{S}},\,B_{i},B^{\dagger}_{i};\, i\in I\}'=\lambda \mathbbm{1}$. The thesis follows from the results derived by Evans. 
\end{proof}
Let us now consider composite quantum systems. For the sake of simplicity, let us consider the case of a total Hilbert space $\mathbb{H}$ made up of $N$ subsystems having the same dimension \emph{i.e.} $\mathbb{H}=\bigotimes_{i=1}^{N}\mathbb{H}_i$, with $\mbox{dim}(\mathbb{H}_i)=d$. In order to prove Theorem \ref{th:secondo}, we have first to generalise the results of Proposition \ref{prop:uno} to this new scenario.\\
\begin{prop}\label{prop:due}
Given the set of operators $\{\bar{B}^{(N)}_{i},\,\bar{B}^{(N)\,\dagger}_i;\,i\,\in[1,\,N]\}$ defined by the Eqs. (\ref{eq:B1})-(\ref{eq:Bn}) and by Eqs.(\ref{eq:ladderoperatorformapp})-(\ref{eq:ladderoperatorformpiuapp}), for any $N < +\infty$ we have
\begin{equation}\label{eq:generalformcommutant}
\{\bar{B}^{(N)}_{i},\,\bar{B}^{(N)\dagger}_i;\,i\,\in[1,\,N]\}'=\lambda\,\mathbbm{1}_{d^N}, \quad \mbox{with}\,\lambda\,\in \mathbb{C}.
\end{equation}  
\end{prop}
\begin{proof} In order to show the validity of (\ref{eq:generalformcommutant}) we proceed by induction. In Proposition \ref{prop:uno} we have shown that the thesis is true for $N=1$. Let assume the thesis true for $N-1$ and let us consider a matrix $X$ in the set of the bounded operators for the total Hilbert space $\mathbb{H}=\bigotimes_{i=1}^{N}\mathbb{H}_{i}$. The matrix $X$ can be decomposed in the following way
\begin{equation}\label{eq:blockdecomposition}
X=\sum_{i,\,j=1}^{d}E_{i,\,j}\otimes\,X_{i,\,j},
\end{equation} 
where $X_{i,\,j}$ is the $d^{N-1}\times d^{N-1}$ matrix of elements $(X_{i,\,j})_{l,\,m}=(X)_{d^{N-1}(i-1)+l,\,d^{N-1}(j-1)+m}$ with $l,\,m\in\,[1,\,d^{N-1}]$.\\
Let us now define the following two sets
\begin{equation}
S_1=\{\bar{B}^{(N)}_{i},\,\bar{B}^{(N)\,\dagger}_i;\,i\,\in[2,\,N]\}
\end{equation}
and
\begin{equation}
S_2=\left\{\bar{B}^{(N)}_{1},\,\bar{B}^{(N)\,\dagger}_1\right\}.
\end{equation}
All the operators in the set $S_1$ are in the form
\begin{equation}\label{eq:decompset1}
\bar{B}^{(N)}_i=\mathbbm{1}_{d}\otimes\bar{B}^{(N-1)}_{i-1},\quad \bar{B}^{(N)\,\dagger}_i=\mathbbm{1}_{d}\otimes\bar{B}^{(N-1)\,\dagger}_{i-1},
\end{equation}
where $\bar{B}^{(N-1)}_{i-1}$ and $\bar{B}^{(N-1)\,\dagger}_{i-1}$ are the Lindblad operators for $(N-1)$-composite problem, while those in the set $S_2$ are in the form
\begin{equation}
\bar{B}^{(N)}_{1}=\bar{B}\otimes \mathbbm{1}_{d^{N-1}},\quad \bar{B}^{(N)\,\dagger}_{1}=\bar{B}^{\dagger}\otimes \mathbbm{1}_{d^{N-1}}.
\end{equation}
By using the following property of the Kronecker product
\begin{equation}\label{eq:kronprop}
(A\,\otimes\,B)(C\,\otimes\,D)=A\,C \otimes\,B\,D,
\end{equation}
we obtain the following expressions for commutators between $X$ and the elements of $S_1$
\begin{equation}\label{eq:commutatorsS1}
\begin{split}
& [X,\,\bar{B}^{(N)}_k]=\sum_{i,\,j=1}^{d}E_{i,\,j}\,\otimes [X_{i,\,j},\,\bar{B}^{(N-1)}_{k-1}],\\
&[X,\,\bar{B}^{(N)\,\dagger}_k]=\sum_{i,\,j=1}^{d}E_{i,\,j}\,\otimes [X_{i,\,j},\,\bar{B}^{(N-1)\,\dagger}_{k-1}].
\end{split}
\end{equation}
If we now set the two commutators in (\ref{eq:commutatorsS1}) to be equal to $0_{d^{N}}$, since the $E_{i,\,j}$ are linearly independent, this happens if and only if \emph{all} the commutators $[X_{i,\,j},\,\bar{B}^{(N-1)}_{k-1}]$ and $[X_{i,\,j},\,\bar{B}^{(N-1)\,\dagger}_{k-1}]$ are equal to $0_{d^{N-1}}$. However, since we have assumed the thesis to be true for $N-1$, this happens if and only if all the blocks $X_{i,\,j}$ are in the form $X_{i,\,j}=\lambda^{(i,\,j)}\,1_{d^{N-1}}$, with $\lambda^{(i,\,j)}\,\in \mathbb{C}$. \\
Therefore, we have that
\begin{equation}\label{eq:partialX}
X=\sum_{i,\,j=1}^{d}E_{i,\,j}\otimes\,\lambda^{(i,\,j)}\,\mathbbm{1}_{d^{N-1}}.
\end{equation}
Let us now consider the commutator of $X$ with the elements of $S_2$. We have
\begin{equation}\label{eq:commtatorsS2piu}
\begin{split}
[X,\bar{B}^{(N)\dagger}_{1}]=&\sum_{l=1}^{d}\sum_{i=2}^{d}(\bar{B}^{\dagger})_{i-1,\,i}\,E_{l,\,i}\otimes \,X_{l,\,i-1}-\\
&-\sum_{l=1}^{d-1}\sum_{i=1}^{d}(\bar{B}^{\dagger})_{l,\,l+1}\,E_{l,\,i}\otimes \,X_{l+1,\,i},
\end{split}
\end{equation}
and 
\begin{equation}\label{eq:commtatorsS2meno}
\begin{split}
[X,\,\bar{B}^{(N)}_{1}]=&\sum_{l=1}^{d}\sum_{i=1}^{d-1}(\bar{B})_{i+1,\,i}\,E_{l,\,i}\otimes\,X_{l,\,i+1}-\\
&-\sum_{l=2}^{d}\sum_{i=1}^{d}(\bar{B})_{l,\,l-1}\,E_{l,\,i}\otimes\,X_{l-1,\,i},
\end{split}
\end{equation}
where we used the explicit form of the operators $\bar{B}$ and $\bar{B}^{\dagger}$ given respectively in Eq. (\ref{eq:ladderoperatorformapp}) and Eq. (\ref{eq:ladderoperatorformpiuapp}), the property (\ref{eq:kronprop}) and the commutation rule in Eq. (\ref{eq:commutatorEiEj}).\\
Compare now Eq.(\ref{eq:commtatorsS2piu}) with Eq. (\ref{eq:commpiu}), and compare Eq.(\ref{eq:commtatorsS2meno}) with Eq. (\ref{eq:commmeno}). These two set of expressions have the same structure, so it is easy to see that requiring Eq.(\ref{eq:commtatorsS2piu}) and Eq.(\ref{eq:commtatorsS2meno}) to be equal to the matrix $0_{d^{N}}$, leads to a set of equations similar to those obtained in the proof of Proposition \ref{prop:uno}. The only difference is that now these equations give a set of constraint for the matrices $X_{i,\,j}$ \emph{i.e.} on the $\lambda^{(i,\,j)}$.\\
Therefore, after a finite number of steps as in Proposition \ref{prop:uno}, we obtain that the matrix $X$ is proportional to $\mathbbm{1}_{d^N}$. For the arbitrariness in the choice of $X$, we conclude that \emph{any} operator $X$ belonging to the commutant $\{\bar{B}^{(N)}_{i},\,\bar{B}^{(N)\,\dagger}_i;\,i\,\in[1,\,N]\}'$ is proportional to $\mathbbm{1}_{d^N}$.\\
In particular, since we have shown in Proposition \ref{prop:uno} that the thesis is true for $N=1$ and here above we have shown that its validity for $N-1$ implies also the validity of the thesis for $N$, by induction principle we have shown that for any $N<+\infty$:\\
\begin{equation}
\{\bar{B}^{(N)}_{i},\,\bar{B}^{(N)\,\dagger}_i;\,i\,\in[1,\,N]\}'=\lambda\,\mathbbm{1}_{d^N}, \quad \mbox{with}\,\lambda\,\in \mathbb{C}.
\end{equation}
\end{proof}
Though here above we have considered the case of a composite system made up of subsystem havenig the same dimension, we stress that the proof of Proposition \ref{prop:due} does not depend on this assumption and its validity can be extended to general case of $N$ subsystems all having different dimension.\\
We are now in the position to prove Theorem \ref{th:secondo}.\\
\newline
\textbf{Theorem 2.}
{\em Given a composite system described by a total Hilbert space $\mathbb{H}$ made up of a finite number of finite dimensional constituents \emph{i.e.} $\mathbb{H}=\bigotimes_{i=1}^{N}\mathbb{H}_{i}$ and $\mbox{dim}(\mathbb{H}_{i})=d_i < \infty$, whose open quantum dynamics is governed by a dynamical semigroup with generator given in the LGKS form, 
if there exists a set of linear complex combinations $\{M_j;\,j\in[1,\cdots,N]\,\}$}
\begin{equation}
M_j=\alpha_{0,\,j} \mathcal{H}_{\mathcal{S}}+\sum_{i\,\in I}\left(\alpha_{i,\,j} B_i+\beta_{i,\,j}  B^{\dagger}_i \right)
\end{equation}
{\em and a unitary operator $U$ such that} 
\begin{equation}
\bar{B}^{(N)}_j=U^{\dagger}M_jU\quad \mbox{or}\quad\bar{B}^{(N)\,\dagger}_{j}=U^{\dagger}M_jU,
\end{equation}
{\em then the dynamical semigroup $\Lambda_t=exp(Lt)$ is irreducible, being the $\{\bar{B}^{(N)}_{j}\}$ specified in eqs.(\ref{eq:B1})-(\ref{eq:Bn}) and being the $\{\bar{B}^{(N)\dagger}_{j}\}$ their adjoint.}\\
\begin{proof}
Let us consider an operator $X$ belonging to the commutant $\{\mathcal{H}_{\mathcal{S}},\,B_{i},B^{\dagger}_{i};\, i\in I\}'$. By definition, it commutes with all the Lindblad operators and the Hamiltonian. As a consequence, it commutes also with their linear combinations $M_j$. By hypothesis, there exists a unitary operator such that
\begin{equation}\label{eq:commMjdue}
[X^U,\,\bar{B}^{(N)}_j]=0\quad \mbox{or}\quad [X^U,\,\bar{B}^{(N)\,\dagger}_j]=0,
\end{equation}
\newline
being $X^{U}=UXU^{\dagger}$ for all $j$. Thanks to Proposition \ref{prop:due}, we have that $X^{U}=\lambda \mathbbm{1}$. Therefore, $X=U^{\dagger}X^{U}U=\lambda \mathbbm{1}$. As a consequence of the arbitrariness in the choice of $X$, it follows that 
\begin{equation}
\{\mathcal{H}_{\mathcal{S}},\,B_{i},B^{\dagger}_{i};\, i\in I\}'=\lambda \mathbbm{1}_{d_{tot}},
\end{equation}
and the dynamical semigroup is irreducible, being $d_{tot}=\prod_{i=1}^{N}\,d_i$.
\end{proof}

\end{document}